\documentclass[a4paper,11pt,,english]{article}

\usepackage[T1]{fontenc}
\usepackage{hyperref}
\usepackage{verbatim}
\usepackage{amsfonts}
\usepackage{amsmath,amssymb,amsthm}
\usepackage{url}
\usepackage[ruled,vlined]{algorithm2e}
\usepackage{fullpage}
\usepackage[usenames]{color}
\usepackage{times}
\usepackage{babel}
\usepackage{enumitem}

\newtheorem{theorem}{Theorem}[section]
\newtheorem{lemma}[theorem]{Lemma}

\newtheorem{corollary}[theorem]{Corollary}

\newtheorem{claim}[theorem]{Claim}
\newtheorem{obs}[theorem]{Observation}
\newtheorem{definition}[theorem]{Definition}

\newcommand{\ignore}[1]{}

\newcommand{\dsum}{\displaystyle\sum}
\newcommand{\expect}{\mathbb{E}}

\newcommand{\eps}{\epsilon}

\newcommand{\mbR}{\mathbb{R}}
\newcommand{\M}{\mathcal{M}}

\newcommand{\I}{\mathcal{I}}
\newcommand{\St}{\mathcal{S}}
\newcommand{\zx}{\mu}

\newcommand{\x}{x}
\newcommand{\win}{\operatorname{win}}
\newcommand{\Gr}{\textsc{Greedy}}

\makeatother

\usepackage{babel}
\makeatother

\begin{document}

\title{ The Secretary Returns}
\author{   Shai Vardi \thanks{Blavatnik School of Computer Science, 
Tel Aviv University.
 E-mail: {\tt  shaivar1@post.tau.ac.il}. This research was supported in part by the Google Europe Fellowship in Game Theory.}
 }
\date{}
\maketitle
\begin{abstract}

In the online random-arrival model, an algorithm receives a sequence of $n$ requests that arrive in a random order. The algorithm is expected to make an irrevocable decision with regard to each request based only on the observed history. We consider the following natural extension of this model:  each request arrives $k$ times, and the arrival order is a random permutation of the $kn$ arrivals; the algorithm is expected to make a decision regarding each request only upon its last arrival. We focus primarily on the case when $k=2$, which can also be interpreted as each request arriving at,  and departing from the system, at a random time.

We examine the secretary problem: the problem of selecting the best secretary when the secretaries are presented online according to a random permutation. We show that when each secretary arrives twice, we can achieve a competitive ratio of $\sim0.768$ (compared to $1/e$ in the classical secretary problem), and that it is optimal. We also show  that without any knowledge about the number of secretaries or their arrival times, we can still hire the best secretary with probability at least $2/3$, in contrast to the impossibility of achieving a constant success probability in the classical setting. 

We extend our results to the matroid secretary problem, introduced by Babaioff et al. \cite{BIK07}, and show a simple algorithm that achieves a $2$-approximation to the maximal weighted basis in the new model (for $k=2$). We show that this approximation factor can be improved in special cases of the matroid secretary problem; in particular, we give a $16/9$-competitive algorithm for the returning edge-weighted bipartite matching problem.
\end{abstract}

\section{Introduction}



The \emph{secretary problem} \cite{Lindley61,Dynkin63} is the following: $n$ random items are presented to an observer in random order, with each of the $n!$ permutations being equally likely. There is complete preference order over the items, which the observer is able to query, for the items he\footnote{We use male pronouns throughout this paper for simplicity. No assumption on the genders of actual agents is intended.} has seen so far. As each item is presented, the observer must either accept it, at which point the process ends, or reject it, and then it is lost forever. The goal of the observer is to maximize the probability that he chooses the ``best'' item (i.e., the one ranked first in the preference order). This problem models many scenarios; one such scenario is the one for which the problem  is named: $n$ secretaries arrive one at a time, and an irrevocable decision whether to accept or reject each secretary is made upon arrival. Another is the house-selling problem, in which buyers arrive and bid for the house, and the seller would like to accept the highest offer.
An alternative way of modeling this problem is the following. Each secretary is allocated, independently and uniformly at random, a real number $r \in [0,1]$, which represents his arrival time. As before, the seller sees the secretaries in the order of arrival and must make an irrevocable decision before seeing the next secretary. It is easy to see that the two models are essentially equivalent (assuming $n$ is known, see e.g., \cite{Bruss84}) 
- the arrival times define a permutation over the secretaries, with each permutation being equally likely.
The optimal solution for the classical secretary problem is well known - wait until approximately $n/e$ secretaries have passed\footnote{The exact number for each $n$ can be computed by dynamic programming, see e.g., \cite{GM66}.} (alternatively until time $t=1/e$), and thereafter, accept a secretary if and only if he is the best out of all secretaries observed so far (e.g., \cite{GM66,Bruss84}). This gives a probability of success 
 of at least $1/e$. 

Consider the following generalization of the secretary problem: Assume that each secretary arrives $k$ times, and the seller has to make a decision upon each secretary's last arrival. 
 We model this  as follows: Allocate each secretary $k$ numbers, independently and uniformly at random from $[0,1)$, which represent his $k$ arrival times. (Equivalently, we may consider only the order of arrivals; in this case each of the $(kn)!$ permutations over the arrival events is equally likely.) A decision whether to accept or reject a secretary must be made between his first and last arrival. We call this problem the \emph{$(k-1)$-returning secretary problem}. The secretary problem is a classical example of the random-arrival online model (e.g., \cite{BMM12, MY11}), and our model immediately applies to this more general framework,
capturing several natural variations thereof, for example:
\begin{enumerate}
\item Requests may not require (or expect) an immediate answer and  will therefore visit the system several times to query it.
\item When requests arrive, the system  gives them either a rejection or an acceptance, or an invitation to return at some later time. It turns out that in many cases, very few requests actually need to return; in the secretary problem, for example, a straightforward analysis shows that the optimal algorithm will only ask $O(\log{n})$ secretaries to return.\footnote{The algorithm will only ask the $i^{th}$ secretary that arrives to return if he is the best out of all the secretaries it has seen thus far. The probability of this is $1/i$. Summing over all secretaries gives the bound.}
\item Requests may enter the system and leave at some later time. The time the request stays in the system can vary from ``until just before the next item arrives'', in which case no information is gained, to ``until the end'', in which case the problem reduces to an offline one. Clearly we would like something in between. When $k=2$, the second random variable allocated to the query can be interpreted  as the time that the query leaves the system, giving a natural formulation  of this property in the spirit of the random-arrival online model.
\end{enumerate}

\subsection{Our Results}
 When each secretary returns once (i.e., $k=2$), we show that the optimal solution has a similar flavor to that of the classical secretary problem - wait until some fraction of the secretaries have passed (ignoring how many times each secretary has arrived), and thereafter hire the best secretary (out of those we have seen so far), upon his second arrival. To tightly bound the probability of success (for large $n$), we examine the case when each of the $2n$ arrival times is selected uniformly at random from $[0,1)$.  We use this model to show that  the success probability tends to $0.76797$ as $n$ grows. In the classical secretary problem, it is essential to know the number of secretaries arriving in order to achieve a constant success probability. We consider the case when $n$ is not known in advance (and there is no extra  knowledge, such as arrival time distribution), and show that by choosing the best secretary we have seen once he returns, with no waiting period, we can still obtain a success probability of at least $2/3$. 
 We also consider cases when $k >2$: we show that for $k=3$, we can achieve a success probability of at least $0.9$, even without knowledge of $n$, and show that setting $k=\Theta(\log{n})$ guarantees success with arbitrarily high probability ($1-\frac{1}{n^{\alpha}}$ for any $\alpha$).
 
We extend our results to    the \emph{matroid secretary problem}, introduced by Babaioff et al., \cite{BIK07}, which is an adaptation of the classical secretary problem to the domain of  \emph{weighted matroids}. A weighted matroid is a pair $\M=(E,\I)$ of elements $E$ and independent sets $\I$, and a weight function $w:E\rightarrow \mbR$, which obeys the properties of heredity and exchange (see Section~\ref{sec:model} for a formal definition).  In the matroid secretary problem, the elements of a weighted matroid are presented in  random order to the online algorithm. The algorithm maintains a set $S$ of selected elements; when an element $e$ arrives, the algorithm must decide whether to add it to $S$, under the restriction that $S \cup \{e\}$ is an independent set of the matroid. The algorithm's goal is to maximize the sum of the weights of the items in $S$. It is currently unknown whether there exists an algorithm that can find a set whose expected weight is a constant fraction of the optimal offline solution. The best result to date is an  $O(\sqrt{\log{\rho}})$-competitive algorithm,\footnote{An online algorithm whose output is within a factor $c$ of the optimal offline output is said to be $c$-competitive; see Section~\ref{sec:matroid} for a formal definition.} where $\rho$ is the rank of the matroid, due to Chakraborty and Lachish \cite{CL12}. We show that in the returning online model, there is an algorithm which is $2$-competitive in expectation (independent of the rank). We also show that for bipartite edge-weighted matching, and hence for transversal matroids\footnote{Transversal matroids (see Section~\ref{sec:ebom} for a definition) are a special case of bipartite edge-weighted matching.} in general, this result can be improved, and show a $16/9$-competitive algorithm.



\subsection{Related Work}
\label{subsec:related}

The origin of the secretary problem is still being debated: the problem first appeared in print in 1960 \cite{Gardner60}; its solution is often credited to Lindley \cite{Lindley61} or Dynkin \cite{Dynkin63}. Hundreds of papers have been published on the secretary problem and variations thereof; for a review, see  \cite{Freeman83}; for a historical discussion, see \cite{Ferguson89}. 
 Kleinberg \cite{Kleinberg05} introduced a version of the secretary problem in which we are allowed to choose $k$ elements, with the goal of maximizing their sum. He gave a $1-O(\sqrt{1/k})$-competitive algorithm, and showed that this setting applies to strategy-proof online auction mechanisms. 

The \emph{matroid secretary problem} was introduced by Babaioff et al., \cite{BIK07}. They gave a $\log{\rho}$-competitive algorithm for general matroids, where $\rho$ is the rank of the matroid, and   several constant-competitive algorithms for special cases of the matroid secretary problem.  Chakraborty and Lachish \cite{CL12} gave an $O(\sqrt{\log{\rho}})$ algorithm for the matroid secretary problem. There have been several improvements on special cases since then. Babaioff et al., \cite{BDGIT09} gave algorithms for the discounted and weighted secretary problems; Korula and P{\'a}l  \cite{KorulaP09} showed that \emph{graphic matroids}\footnote{In a \emph{graphic matroid} $G=(V,E)$, the elements are the edges of the graph $G$ and a set is independent if it does not contain a cycle.} admit $2e$-competitive algorithms;
 Kesselheim et al., \cite{KRTV13} gave a $1/e$-competitive algorithm for the secretary problem on transversal matroids and showed that this is optimal.   Soto \cite{Soto13} gave a $2e^2/(e-1)$-competitive algorithm when the adversary can choose the set of weights of the elements, but the weights are assigned at random to the elements, (and the elements are presented in a random order). Gharan and Vondr{\'a}k \cite{GharanV13} showed that once the weights are random, the ordering can be made adversarial, and that this setting still admits $O(1)$-competitive algorithms.
There have been other interesting results in this field; for a recent survey, see  \cite{Dinitz13}.

\subsection{Comparison with Related Models}
There are several other ther papers which consider online models with arrival and departure times. Due to the surge in interest in algorithmic game theory over the past 15 years, and the economic implications of the topic, it is unsurprising that many of these papers are economically motivated. 
 Hajiaghayi et al., \cite{HKP04}, consider the case of an auction in which an auctioneer has $k$ goods to sell and the buyers arrive and depart dynamically. They notice and make use of the connection to the secretary problem to design strategy-proof mechanisms: they design an $e$-competititive (w.r.t. efficiency) strategy-proof mechanism for the case $k=1$, which corresponds to the secretary problem, and extend the results to obtain $O(1)$-cometitive mechanisms for $k>1$. Hajiaghayi et al., \cite{HKMP05}, design strategy-proof mechanisms for online scheduling in which agents bid for access to a re-usable resource such as processor time or wireless network access, and each agent is assumed to arrive and depart dynamically.  Blum et al., \cite{BSZ06}, consider online auctions, in which a single commodity is bought by multiple buyers and sellers whose bids arrive and expire at different times. They present an $O(\log{(p_{max}-p_{min})})$-competitive algorithm for maximizing profit and an $O(\log({p_{max}/p_{min}}))$-competitive algorithm for  maximizing volume where the bids are in the range $[p_{min}, p_{max}]$, and a strategy-proof algorithm for maximizing social welfare. They also show that their algorithms achieve almost optimal competitive ratios. Bredin and Parkes \cite{BP12} consider online \emph{double auctions}, which are matching problems with incentives, where agents arrive and depart dynamically. They show how to design strategy-proof mechanisms for this setting.

 \subsection{Paper Organization}
 In Section~\ref{sec:model} we introduce our model. In Section~\ref{sec:classic_return} we provide an optimal algorithm for the returning secretary problem. In Section~\ref{sec:matroid} we give a $2$-competitive algorithm for the returning matroid secretary problem, and in Section~\ref{sec:ebom}, we show we can improve this competitive ratio to $16/9$ for transversal matroids (and more generally, returning edge-weighted bipartite matching). In Appendices~\ref{app:3sec} and~\ref{app:logn}, we analyze the cases of the $k$-returning secretary problem for $k=3$ and $k = \Theta(\log{n})$.

\section{Model and Preliminaries}
\label{sec:model}
Consider the following scenario. There are $n$ items which arrive in an online fashion, and each item arrives $k$ times. Each arrival of an item is  called a \emph{round}; there are  $kn$ rounds. The order of arrivals is selected uniformly at random from the $(kn)!$ possible permutations. An algorithm observes the items as they arrive, and must make an irrevocable decision about each item upon the item's last appearance. We call such an algorithm a $(k-1)$-\emph{returning online algorithm} and the problem it solves a $(k-1)$-\emph{returning online problem}. Because the problem is most natural when $k=2$, for the rest of the paper, (up to and including Appendix~\ref{sec:ebom}), we assume that $k=2$ (and instead of  ``$1$-returning'', we simply say ``returning''.) In Appendices~\ref{app:3sec} and~\ref{app:logn} we consider scenarios when $k>2$.

 We use the following definition of matroids:
\begin{definition}
A  matroid $\M=(E, \I)$ is an ordered pair, where $E$ is a finite set of elements (called the \emph{ground set}), and $\I$ is a family of subsets of $E$,  (called the \emph{independent sets}), which satisfies the following properties:
\begin{enumerate}[ noitemsep]
\item $\emptyset \in \I$,
\item If $X \in \I$ and $Y \subseteq X$ then $Y \in \I$,\label{prop2}
\item If $X,Y \in \I$ and $|Y|< |X|$ then there is an element $e\in X$ such that $Y \cup \{e\} \in \I$.\label{prop3}
\end{enumerate}
\end{definition}
Property~\eqref{prop2} is called the \emph{hereditary property}. Property~\eqref{prop3} is called the \emph{exchange property}. An independent set that becomes dependent upon adding any element of $E$ is called a \emph{basis} for the matroid.
In a \emph{weighted matroid}, each element $e \in E$ is associated with a weight $w(e)$.
The \emph{returning matroid secretary problem} is the  following: Each element of a weighted matroid $\M=(E,\I)$ arrives twice, in an order selected uniformly at random out of the $(2n)!$ possible permutations of arrivals. The algorithm maintains a set of selected elements, $S$, and may add any element to $S$ at any time between (and including) the first and second appearances of the element, as long as $S \cup \{e\} \in \I$. The goal of the algorithm is to maximize the sum of the weights of the elements in $S$. The success of the algorithm is defined by its \emph{competitive ratio}.
\begin{definition}[competitive ratio, $c$-competitive algorithm]
If the weight of a maximal-weight basis of a matroid is at most $c$ times the  expected weight of the set selected by an algorithm (where the expectation is over the arrival order), the algorithm is said to be $c$-\emph{competitive}, and its \emph{competitive ratio} is said to be $c$.
\end{definition}
A special case of the returning matroid secretary problem is the \emph{returning secretary problem}, in which there are $n$ secretaries, each of whom arrives twice. The goal of the algorithm is to identify the best secretary. The algorithm is \emph{successful} if and only if it chooses the best secretary, and we quantify how ``good'' the algorithm is by its success probability.

Without loss of generality, we assume throughout this paper that the weights of all the elements are distinct.\footnote{Babaioff et al., \cite{BIK07} show that we do not lose generality by this assumption in the matroid secretary problem. The result immediately applies to our model.} Although we do not discuss computational efficiency in this work, all the algorithms in this paper are polynomial in the succinct representation of the matroid. 

We denote the set $\{1, 2, \ldots n\}$ by $[n]$.

\section{The Returning Secretary}
\label{sec:classic_return}

Assume that there are $n$ secretaries that arrive in an online fashion.  Each secretary arrives twice, and the order is selected uniformly at random from the $(2n)!$ possible orders.  At all times, we keep note of who the best secretary is out of all the secretaries seen so far. We call this secretary the \emph{candidate}. That is, in each round,  if the secretary that arrived is better than all other secretaries that arrived before this round, he becomes the candidate. Note that it is possible that in a given round,  the candidate will have already arrived twice. At any point between each  secretary's first and second arrival, we can  \emph{accept} or \emph{reject} him; an acceptance is final, a rejection is only final if made upon the second arrival. Once we accept a secretary, the process ends. 
We \emph{win} if we accept (or \emph{choose}) the best secretary. We would like to maximize the probability of winning.
 
\subsection{Optimal Family of Rules}
 
What is the best strategy for maximizing the probability of winning? We first show that the optimal rule must be taken from the family of stopping rules as described in the following lemma.

\begin{lemma}\label{lemma:zx}
The optimal strategy for choosing the best secretary in the returning secretary problem has the following structure: wait until $d$ distinct secretaries have arrived; thereafter,  accept the best secretary out of the secretaries seen so far, when he returns. 
\end{lemma}

\begin{proof} Without loss of generality, we can restrict our attention to strategies that make decisions regarding a secretary $s$ only upon  $s$'s arrivals, as every strategy that makes decisions between the two arrival times has an equivalent strategy that defers the decision making to the second arrival.
Let $d_r$ be the random variable denoting the number of distinct secretaries that have arrived up to (and including) round $r$ ($r \in [2n]$). Denote by $H(r) = \{x_1, x_2, \ldots, x_{r-1}\}$ the history at round $r$, where $x_i = (y_i, z_i)$: $y_i$ is the relative rank (among the secretaries that have arrived until now) of the secretary that arrived at time $i$, and $z_i$ represents whether this is the first or second time that this secretary has arrived (i.e. $y_i \in [d_r]$, $z_i \in \{1,2\}$). Any (deterministic) strategy $\St$ must have the following structure: for every realization of $x_r = (y_r, z_r)$, and $H(r)$,  $\St$ must accept or reject. That is $S: (H_r, y_r, z_r) \rightarrow \{\mathrm{accept, reject}\}$. Denote the optimal strategy by $\St^*$. Clearly,
\begin{enumerate}[noitemsep, nolistsep]
\item If the $t^{th}$ secretary is not the best, we will not choose him: $\forall y_r \neq 1$, $\St^*(H_r, y_r, z_r) = \mathrm{ reject}$.
\item If this is the first time we have seen a secretary, we cannot gain anything by choosing him now. It is better to wait for the second arrival, as we lose nothing by waiting: $\St^* (H_r, y_r, 1)= \mathrm{reject}$.
\end{enumerate} 
Therefore, we only need to consider choosing the best secretary we have seen so far when he returns; i.e., we only accept at time $t$ such that $y_r = 1, z_r =2$. For all other values of $y_i$ and $z_i$, $\St^*$ must reject; henceforth, we only focus on the case that $y_r = 1, z_r =2$, and omit this from the notation.
Denote the event that $\St^*$ accepts on history $H_r$ by $Acc(H_r)$.
As $\St^*$ is a probability-maximizing strategy,  
\begin{equation}
\St^* (H_r) = \mathrm{accept} \iff \Pr[\win|Acc(H_r)] \geq \Pr[\win|\neg Acc(H_r)].\label{inequ}
\end{equation}

Given that $d_r = d$, $\Pr[\win|Acc(H_r)] = d/n$, as this is exactly the probability that the best secretary is part of a group of $d$ secretaries selected uniformly at random. 
Although we cannot give such an elegant formula for $\Pr[\win|\neg Acc(H_r)]$, we know that it is the probability of winning given that we have seen $d$ secretaries, rejected them all, and have $(n-d)$ secretaries remaining to observe; hence, the probability is dependent only on $d$ (as $n$ is fixed). Denote this probability function by $f(d)$. We do not attempt to describe $f$, other than to say that $f$ must be non-increasing in $d$. (This is easy to see: $f(d)\geq f(d+1)$ as a possible strategy is to always reject the $d^{th}$ secretary.)

As the left side of~\eqref{inequ}  is an increasing function of $d$, and the right side is a decreasing function  of $d$ (and as $\St^*$ is a probability-maximizing function), $\St^*$ will accept only if the number of distinct secretaries that have arrived is at least $d^*$, the minimal $d$ such that $d/n \geq f(d)$. We can conclude that the optimal strategy is to observe  the first $d^*$ secretaries without hiring any and to choose the first suitable secretary thereafter. It is easy to see (similarly to \cite{Bruss2000}), that randomization cannot lead to a better stopping rule.
\end{proof}

From Lemma~\ref{lemma:zx}, we can conclude that there is some function $f:n\rightarrow [0,n]$ for which the optimal  algorithm for the returning secretary problem is Algorithm~\ref{alg:sec}. 

\vspace{10pt}
\begin{algorithm}[H]
\caption{Returning secretary algorithm with function $f$}\label{alg:sec}
\SetKwInOut{Input}{Input}\SetKwInOut{Output}{Output}

\Input{$n$, the number of secretaries }
\Output{A secretary $s$}
\BlankLine
the Candidate = $\emptyset$\;

\For{round $r=1$ to $2n$}{Let $i_r$ be the secretary that arrives on round $r$\;
Denote by $d_r$ the distinct number of secretaries that have arrived up to round $r$\;
\If {$i_r$ is the Candidate}{\If{$d_r > f(n)$}{Return $i_r$\;}}
\If {$i_r$ is better than the Candidate}{the Candidate = $i_r$\;}}

\end{algorithm}
 \vspace{10pt}

We do not, at this time, attempt to find the function $f$ for which  Algorithm~\ref{alg:sec} is optimized;  we will optimize the parameter of a similar algorithm for a slightly different setting in Subsection~\ref{subsection:opt}. For now, we focus on the special case where $f(n)\equiv 0$, which we call the \emph{no waiting} case. Aside from being interesting in their own right, these results will come in useful later on, for tightly bounding the success probability.

\subsection{The No Waiting Case}
\label{sub_twothir}
In the classical secretary problem, even if we don't know $n$ in advance, we can still find the best secretary with a reasonable probability, assuming we have some other information regarding the secretaries. For example, the secretaries can have an known arrival time density over $[0,1]$ \cite{Bruss84}\footnote{Note that this is different from the alternative formulation described in the introduction as in this case $n$ is unknown.};
$n$ can be selected from some known distribution \cite{PS72};  there are other, similar scenarios (see e.g., \cite{Stewart81, ABT82, Por87}). However, with no advance knowledge at all,  it is impossible to attain a success probability better  than $1/n$ (with a deterministic algorithm):  if we don't accept the first item, we run the risk of there being no other items, while if we do accept it, we have accepted the best secretary with probability $1/n$. It is easy to see that while randomization may help a little, is cannot lead to a constant success probability.
In the returning-online scenario, though, we have the following result.
\begin{theorem}\label{thm:23}
In the returning secretary problem, even if we have no previous information on the secretaries, including the number of secretaries that will arrive, we can hire the best secretary with probability at least $2/3$. 
\end{theorem}

Denote by \emph{$\win$} the event that we hire the best secretary. Theorem~\ref{thm:23} is immediate from the following lemma.

\begin{lemma}\label{lemma:NW}
When applying Algorithm~\ref{alg:sec} to the returning secretary problem with $f(n)\equiv0$, 
\begin{equation*}
\Pr[\win] = \frac{2n+1}{3n}.
\end{equation*}
\end{lemma}

\begin{proof}
Let us call the best secretary Don.  If we reach round $i$ and see Don, we say we \emph{win} on round $i$, and denote this event $\win_i$.  (Notice that we can say  that we win at this point even though this is the first time we see Don, as we will certainly hire him). The probability of winning on round $1$ is exactly the probability that Don arrives first:
$$\Pr[\win_1] = \frac{2}{2n}.$$
We win on round $2$ if any secretary other than Don arrived on round $1$, and Don arrived on round $2$.
$$\Pr[\win_2] =  \left(\frac{2n-2}{2n}\right)\left(\frac{2}{2n-1}\right).$$
The probability of winning on round $i>2$ is the following (the best secretary we had seen until that point did not return between rounds $2$ and $i-1$, and Don arrived on round $i$):
$$\Pr[\win_i] = \left(\frac{2n-2}{2n}\right)\left(\frac{2n-4}{2n-1}\right)\left(\frac{2n-5}{2n-2}\right)\left(\frac{2n-6}{2n-3}\right)\dots \left(\frac{2n-i-1}{2n-i+2}\right) \left(\frac{2}{2n-i+1}\right).$$
 Therefore

\begin{align}
\Pr[\win] &= \frac{1}{n} + \frac{1}{n(2n-1)(2n-3)}\dsum_{i=2}^{2n-2} (2n-i)(2n-i-1) \notag\\
&=\frac{1}{n}+ \frac{2(n-1)(2n-1)(2n-3)}{3n(2n-1)(2n-3)}\label{eq:simple1}\\
&=\frac{3}{3n}+ \frac{2(n-1)}{3n} \notag\\
&=\frac{2n+1}{3n},\notag
\end{align}
 where~\eqref{eq:simple1} is reached by substituting $j=2n-i$ and simplifying the sum.
\end{proof}

\subsection{Optimizing the Success Probability}
\label{subsection:opt}
We would now like to optimize $f$ in Algorithm~\ref{alg:sec} in order to maximize the algorithm's success probability.
For ease of analysis, we turn to the alternative model for the secretary problem: instead of generating a random permutation over the secretaries, each secretary $i$ is allocated, uniformly and independently at random, two real numbers $r^1_i, r^2_i \in [0,1)$, representing his two arrival times, $i^1$ and $i^2$. Assume that $f^*$ is the optimal function for Algorithm~\ref{alg:sec}. Fix $n$ and let $\zx$ denote the time of the arrival of the $(f^*(n))^{th}$ distinct secretary. It is easy to see  that  the two models are asymptotically  identical: for large $n$, $\Pr[i^j \text{is one of  the first }f^*(n) \text{ arrivals}]\approxeq \Pr[i^j \in [0,\zx)] $. The analysis in this model is much cleaner, and so, for simplicity, (and at the expense of accuracy for small $n$), we use it to obtain our bounds. The optimal algorithm for the returning secretary problem in this model is  Algorithm~\ref{alg:poop}.

We introduce some new notation.  
\begin{itemize} 
\item Denote by $\win(\zx)$ the event that we hire the best secretary when using Algorithm~\ref{alg:poop} with parameter $\zx$.
\item Let $\alpha_i(\zx)$ be the event that $r^1_i, r^2_i \in [0,\zx)$. 
\item Let $\beta_i(\zx)$ be the event that $r^1_i \in [0,\zx)$ and $r^2_i \in [\zx,1)$ or vice versa.
\item Let $\gamma_i(\zx)$ be the event that $r^1_i, r^2_i \in [\zx,1)$.
\end{itemize} 
We omit $(\zx)$ from the notation when it is clear from context. Label the best  secretary by $1$, the second best by $2$ and so on. Denote by $\win(NW_i)$ the event that we find the best secretary in the no waiting scenario with $i$ secretaries (recall that this is $\frac{2i+1}{3i}$). We make the following observations, which depend on the arrival times being independent.

\vspace{10pt}
\begin{algorithm}[H]
\caption{Returning secretary algorithm with parameter $\zx \in [0,1)$}\label{alg:poop}
\SetKwInOut{Input}{Input}\SetKwInOut{Output}{Output}

\Output{A secretary $s$}
\BlankLine
the Candidate = $\emptyset$\;
Observe the first secretary\;
\While{there are secretaries that have not arrived}{Let $i$ be the observed secretary\;
Let $t_i$ be the time that $i$ is observed\;
\If {$i$ is the Candidate}{\If{time $\geq\zx$} {Return $i$\;}}
\If {$i$ is better than the Candidate}{the Candidate = $i$\;}
Observe the next secretary\;}
\end{algorithm}
\vspace{10pt}

\begin{obs}\label{obs:abc}
$\forall i \in [n], \Pr[\alpha_i(\zx)] = \zx^2, \Pr[\beta_i(\zx)] = 2\zx(1-\zx), \Pr[\gamma_i(\zx)] = (1-\zx)^2$.
\end{obs}
\begin{obs}\label{obs:abc2}
$\Pr[\win|\gamma_1, \gamma_2, \ldots, \gamma_i, \alpha_{i+1}] = \Pr[\win(NW_i)]$.
\end{obs}
\begin{proof}
If $\gamma_1, \gamma_2, \ldots, \gamma_i$ hold then all of the  appearances of the best $i$ secretaries are in the interval $[\zx,1)$. Both appearances of the  $(i+1)^{th}$ best secretary are in $[0,\zx)$; therefore we will definitely choose one of the $i$ best secretaries, and the probability of choosing the best is as in the no waiting scenario.
\end{proof}
\begin{obs}\label{obs:abc3}
$\Pr[\win|\gamma_1, \gamma_2, \ldots, \gamma_{i}, \beta_{i+1}] = \Pr[\win(NW_{i+1})| \text{secretary } i+1 \text{ is the first to arrive}]$.
\end{obs}
\begin{proof}
If $\gamma_1, \gamma_2, \ldots, \gamma_i$ and $\beta_{i+1}$ hold then all appearances of the best $i$ secretaries are in the interval $[\zx,1)$, and the $(i+1)^{th}$ secretary arrived once by time $\zx$. This reduces to the problem of choosing the best secretary in the no waiting scenario, given that the $(i+1)^{th}$ secretary arrives first.
\end{proof}
\begin{claim}\label{claim:abc}
$\Pr[\win(NW_{i+1})| \text{secretary } i+1 \text{ is the first to arrive}]= \frac{2i}{2i+1}\Pr[\win(NW_{i})]$.
\end{claim}
\begin{proof}
Given that $i+1$ is the first to arrive, if $i+1$ arrives second, we lost. If not, $i+1$ cannot be chosen anymore, and we are exactly in the no waiting scenario with $i$ secretaries. The probability that $i+1$ arrives second given that he also arrives first is $\frac{1}{2i+1}$.
\end{proof}

Combining Observation~\ref{obs:abc3} and Claim~\ref{claim:abc} gives the following corollary.
\begin{corollary}\label{corr:abc2}
$\Pr[\win|\gamma_1, \gamma_2, \ldots, \gamma_{i}, \beta_{i+1}] = \frac{2i}{2i+1}\Pr[\win(NW_{i})]$.
\end{corollary}

We are now able to obtain  a recursive representation of $\Pr[\win|\gamma_1, \gamma_2, \ldots, \gamma_i]$.
\begin{claim}\label{claim:abc2}
$\Pr[\win|\gamma_1, \gamma_2, \ldots, \gamma_i] =  \frac{\zx^2 + 4\zx i - 2\zx^2i}{3i} + (1-\zx)^2 \Pr[\win|\gamma_1, \gamma_2, \ldots,  \gamma_{i+1}]$.
\end{claim}
\begin{proof}
\begin{align}
\Pr[\win|\gamma_1, \gamma_2, \ldots, \gamma_i] =& \Pr[\win|\gamma_1, \gamma_2, \ldots, \gamma_i, \alpha_{i+1}]\Pr[\alpha_{i+1}] \notag\\
 &+\Pr[\win|\gamma_1, \gamma_2, \ldots, \gamma_i, \beta_{i+1}]\Pr[\beta_{i+1}] \notag\\
 &+\Pr[\win|\gamma_1, \gamma_2, \ldots, \gamma_i, \gamma_{i+1}]\Pr[\gamma_{i+1}] \notag\\
 =& \zx^2\Pr[\win(NW_{i})] +\frac{4\zx i(1-\zx)}{2i+1}\Pr[\win(NW_{i})] \label{eq:abc}\\
 &+ (1-\zx)^2 \Pr[\win|\gamma_1, \gamma_2, \ldots,  \gamma_{i+1}] \notag\\
 =& \frac{\zx^2 + 4\zx i - 2\zx^2i}{2i+1}\Pr[\win(NW_{i})] + (1-\zx)^2 \Pr[\win|\gamma_1, \gamma_2, \ldots,  \gamma_{i+1}], \notag\\
 =& \frac{\zx^2 + 4\zx i - 2\zx^2i}{3i} + (1-\zx)^2 \Pr[\win|\gamma_1, \gamma_2, \ldots,  \gamma_{i+1}],\label{eq:2i}
\end{align}
where~\eqref{eq:abc} is due to Observations~\ref{obs:abc}, and~\ref{obs:abc2}  and Corollary~\ref{corr:abc2}, and~\eqref{eq:2i} is due to Lemma~\ref{lemma:NW}.
\end{proof}

 \begin{claim} For any constant $k$, and any $\zx \in [0,1)$,
 \begin{equation}
 \Pr[\win] \geq 2\zx(1-\zx) + \dsum_{i=1}^k \left( \frac{(1-\zx)^{2i}(\zx^2+4\zx i-2\zx^2i)}{3i}\right) + \frac{2}{3}(1-\zx)^{2k+1} .\label{eq:lemma}
 \end{equation} 
 \end{claim}
\begin{proof}

\begin{align}
\Pr[\win] =& \Pr[\win|\alpha_1]\cdot\Pr[\alpha_{1}] +\Pr[\win| \beta_{1}]\cdot\Pr[\beta_{1}] +\Pr[\win|\gamma_1]\cdot\Pr[\gamma_{1}] \notag\\
=& 0\cdot (\zx^2)+ 1\cdot 2\zx(1-\zx) +  \Pr[\win|\gamma_1] \cdot (1-\zx)^2\label{eq:obs34},
\end{align}
where~\eqref{eq:obs34} is due to Observation~\ref{obs:abc}. 

Recursively applying Claim~\ref{claim:abc2}, and noticing that $\Pr[\win|\gamma_1, \gamma_2, \ldots, \gamma_i] \geq \frac{2}{3}$, for all $i$, completes the claim.

\end{proof}

\begin{lemma}
For any $x \in [0,1)$, $\Pr[\win] \geq 2x -\frac{4}{3}x^2 - \frac{1}{3}(1-x)^2\log(1-x^2)$.
\end{lemma}
\begin{proof}
Substituting $\x=1-\zx$ in~\eqref{eq:lemma}, and ignoring the lowest order term, we get
 \begin{align}
 \Pr[\win] &\geq 2x(1-x) + \frac{1}{3}\dsum_{i=1}^k x^{2i}\left( \frac{(1-x)^2+4(1-x) i-2(1-x)^2 i}{i}\right) \notag\\
 &= 2x(1-x) + \frac{1}{3}(1-x)^2\dsum_{i=1}^k \frac{x^{2i}}{i} +  \frac{1}{3}\dsum_{i=1}^k x^{2i}(4-4x) -2(1-x)^2)  \notag\\
  &= 2x(1-x) + \frac{1}{3}(1-x)^2\dsum_{i=1}^k \frac{x^{2i}}{i} +  \frac{1}{3}\dsum_{i=1}^k x^{2i}(2-2x^2) \notag\\
   &= 2x(1-x) + \frac{1}{3}(1-x)^2\dsum_{i=1}^k \frac{x^{2i}}{i} +  \frac{1}{3}\left( \dsum_{i=1}^k 2x^{2i} -  \dsum_{i=1}^k 2x^{2(i+1)}\right)  \notag\\
   &\geq 2x(1-x) + \frac{1}{3}(1-x)^2\dsum_{i=1}^k \frac{x^{2i}}{i} +\frac{2}{3}x^2 \label{eq:ignore}\\
   &\underset{k\rightarrow \infty}{\longrightarrow} 2x -\frac{4}{3}x^2 - \frac{1}{3}(1-x)^2\log(1-x^2),\label{eq:taylor}
\end{align}
 where in~\eqref{eq:ignore}, we once again ignore the lowest order term, and~\eqref{eq:taylor} is because $\dsum_{i=1}^{\infty}\frac{y^{i}}{i}$ is the Taylor series for $-\log(1-y)$, for $|y|<1$.

\end{proof}
 Differentiating~\eqref{eq:taylor}, we find that the winning probability is maximized at $x= \sqrt{\frac{e^5-e^{W(2 e^5)}}{e^{5/2}}} \approx 0.727374$, where $W(z)$ is the product log function.
 This implies $\zx \approx 0.272626$, and for this value, $\Pr[\win] \approx 0.767974$. This gives our main result of the section.
\begin{theorem}
The optimal algorithm for the returning secretary problem, for $n \rightarrow \infty$, is Algorithm~\ref{alg:poop} with  $\zx = 0.272626$;  the probability of hiring the best secretary is at least $ 0.76797$.
\end{theorem}
It is interesting to note that $\zx = 0.272626$ implies that $f(n) \underset{n \rightarrow \infty}\longrightarrow 0.4709n$ (as the expected number of secretaries that arrive in $[0,x)$ is $x^2 +2x(1-x)$), meaning that in the optimal strategy, we should wait until we have seen almost half of the secretaries before considering hiring.


\section{The Returning Matroid Secretary}
\label{sec:matroid}

We show that in the returning online model, when $k=2$,  a simple algorithm obtains a $2$-approximation to the maximum-weight basis of the matroid. 
It is a well known property of matroids (e.g., \cite{Rado57}), that the Greedy algorithm always finds a maximum-weight basis. Algorithm~\ref{alg:matroid}, in essence,  lets $n$ elements arrive, and then runs the Greedy algorithm on the elements which have only arrived once.
\vspace{10pt}

\begin{algorithm}[H]
\caption{Returning matroid secretary algorithm}\label{alg:matroid}
\SetKwInOut{Input}{Input}\SetKwInOut{Output}{Output}

\Input{a cardinality $n=|E|$ of the matroid $\M=(E,\I)$}
\Output{an independent set $S \in \I$}
\BlankLine
Let $n$ elements arrive, without choosing any element\;
Let $E'$ denote the elements which only arrived once thus far\;
Relabel the elements of $E'$ by $1, 2, \ldots, |E'|$, such that $w_1\geq w_2 \geq \cdots \geq w_{|E'|}$\;
$S \leftarrow \emptyset$\;
\For{$i = 1$ to $|E'|$}{
	\If{$S \cup i \in \I$}{
	$S \leftarrow S \cup i$\;
	}
}
Return $S$\;
\end{algorithm}
\vspace{10pt}

\begin{theorem}\label{thm:matroidmain}
Algorithm~\ref{alg:matroid} is $2$-competitive in expectation. 
\end{theorem}
To prove Theorem~\ref{thm:matroidmain}, we need Claims~\ref{claim:xe} and~\ref{claim:grb}.
Let  $X_e$ be a random variable which is $1$ iff $e$ appears in $E'$.
\begin{claim} \label{claim:xe} For any $e \in E$,
 $$\Pr[X_e = 1] = \frac{n}{2n-1}.$$
\end{claim}

\begin{proof}
Let $S$ denote the set of the first $n$ elements to arrive and $T$ denote the set  of the last $n$ elements. For element $e$, denote its two arrivals by $e_1$ and  $e_2$.
$e \in E'$ iff $e_1 \in S$ and $e_2 \in T$ or vice versa.
 \begin{align*}
 \Pr[e_1 \in S] &= \frac{1}{2},\\
\Pr[e_2 \in T|e_1 \in S]&= \frac{n}{2n-1},\\
Pr[e_1 \in S \wedge e_2 \in T] &= \Pr[e_1 \in S] \Pr[e_2 \in T|e_1 \in S] = \frac{n}{4n-2}.
 \end{align*}
By symmetry, $Pr[e_1 \in T \wedge e_2 \in S] = \frac{n}{4n-2}$, and by summation (as these events are disjoint), the claim follows.
\end{proof}
Let $B^*$ be a maximum-weight basis of $\M$, and $\Gr(E')$ be the output of the Greedy algorithm on the set $E'\subseteq E$. Denote the restriction of $B^*$ to $E'$ by $B^*|E'$. That is, $B^*|E'$ is the elements of $E'$ which appear in $B^*$. 
\begin{claim}\label{claim:grb}
Let $\Gr(E') =\{e_1, e_2, \ldots, e_{\ell}\}$ and $B^*|E' = \{f_1, f_2, \ldots, f_k\}$, where $\Gr(E')$ and  $B^*|E'$ are sorted by weights in a non-increasing order. That is $w(e_1) \geq w(e_2)\geq \ldots \geq w(e_{\ell})$, and $w(f_1)\geq w(f_2)\geq\ldots\geq w(f_k)$. Then  for $1 \leq i \leq k$, $e_i \geq f_i$.
\end{claim}
\begin{proof}
First, notice that by the hereditary property, $B^*|E \in \I$. Because the Greedy algorithm necessarily finds an independent set, by the exchange property, it must hold that $\ell \geq k$.
Now assume by contradiction that there exists some $i$ for which $f_i>e_i$. Let $j$ be the first index for which this inequality holds. It holds that $$e_{j-1} \geq f_{j-1} \geq f_j > e_j.$$
Consider the subsets $X=\{f_1, f_2,\ldots,f_j\}$ and $Y = \{e_1, e_2,\ldots, e_{j-1}\}$. By the exchange property, one of the elements of $X$ can be added to $Y$ while maintaining the independence property. But the Greedy algorithm chose $e_j$, which is lighter than all of the elements of $X$, a contradiction.  
\end{proof}

\begin{proof}[Proof of Theorem~\ref{thm:matroidmain}]
Denote the expected weight of the set returned by Algorithm~\ref{alg:matroid} by $W$. 
\begin{align}
W &= \dsum_{E' \subset E} w(\Gr(E')) \Pr[E'] \notag\\
&\geq   \dsum_{E' \subset E} w(B|E') \Pr[E'] \label{eq2}\\
& = \dsum_{E' \subseteq E}\dsum_{e \in B} w(e) \Pr[e \in E'] \notag\\
& =  \dsum_{e \in B} w(e) \dsum_{E' \subseteq E} \Pr[e \in E'] \notag\\
& =  \dsum_{e \in B} w(e)  \Pr[X_e =1] \notag\\
& = \frac{n}{2n-1}w(B^*), \label{eq:3}
\end{align}
where~\eqref{eq2} and \ref{eq:3} are due to Claims~\ref{claim:grb} and \ref{claim:xe} respectively.
\end{proof}

Algorithm~\ref{alg:matroid} lets $n$ elements arrive, and runs the Greedy algorithm. In expectation, though, only $3/4$ of the elements have already arrived. Therefore, there are roughly $1/4$ of the elements which the algorithm doesn't even make an attempt to add to the set. Clearly, if it continued to add elements in a Greedy fashion, this could only improve the competitive ratio. However, we show that this will not be much help in the general case. We use  the example given by Babaioff et al., \cite{BIK07}: There is a  graphic matroid $G=(V,E)$, where $V=\{u,v,w_1, w_2,\ldots, w_m\}$ and $E=(u,v) \cup \{(u,w_i), (v, w_i)|i=1,2,\ldots, m\}$. The weights are assigned to the edges  as follows:
$w(u,v) = m+1, w(u,w_i)\in (\eps, 2\eps), w(v,w_i) \in (2\epsilon, 3\eps)$, for some small constant $\eps>0$. If $e^*$ was not added to the forest in the Greedy phase just after the first $n$ rounds, the probability it will be added is exponentially small in $m$, as w.h.p. there is a pair of edges $(u,w_j), (v,w_j)$ that has been added to the forest. If $e^*$ is not added to the forest, the competitive ratio is at most $\frac{1}{3\eps}$, and so continuing to use Greedy after $n$ elements have passed will not significantly improve the competitive ratio.

In the case of transversal matroids though, we can use local improvements  to improve the competitive ratio to $16/9$.

\section{Returning bipartite edge-weighted matching} 
\label{sec:ebom}

The \emph{returning bipartite edge-weighted matching} problem is a generalization of the returning transversal matroid problem.\footnote{A \emph{transversal matroid} is  a bipartite graph $G=(L\cup R, E)$ where the elements are the vertices of $L$ and the independent sets are sets of endpoints of matchings in the graph. Transversal matroids are a special case of  bipartite edge-weighted matching, in which all the edges incident on the same vertex $\ell \in L$ have the same weight.}
Let $G=(L \cup R, E)$ be a bipartite graph with a weight function $w:E \rightarrow \mbR^+$. We are initially given $R$ and $n=|L|$. In each step, a vertex $v \in L$ arrives together with its edges (and the edges' weights). Each vertex arrives twice, and the order of arrival is selected uniformly at random from the $(2n)!$ possible arrival orders.   When a vertex $\ell \in L$ arrives for the second time, it is either matched to one of the free vertices in $R$ that are adjacent to $\ell$, or left unmatched. The goal of the algorithm is to maximize the weight of the matching. Note that if $|R|=1$, and we succeed only if we find the maximum matching, this is exactly the returning secretary problem. 

We present a variation on the returning matroid secretary algorithm, where instead of the Greedy algorithm, we use a maximum-matching algorithm (using any maximum matching algorithm, e.g., the Hungarian method \cite{Kuhn55}). We then use local improvements, similarly to \cite{KRTV13}.

\vspace{10pt}
\begin{algorithm}[H]
\caption{Returning bipartite edge-weighted matching algorithm}\label{alg:ebom}
\SetKwInOut{Input}{Input}\SetKwInOut{Output}{Output}

\Input{vertex set $R$ and a cardinality $n=|L|$} 
\Output{a matching $M$}
\BlankLine
Let $L_r$ be the vertices that arrived until round $r$.\;
Let $L' \subset L_n$ denote the vertices that only arrived once until round $n$\;
$M = \text{optimal matching on } G[L' \cup R]$\;
\For{each subsequent round $t>n$, when vertex $\ell_t \in L$ arrives}
{$M_{t} = \text{optimal matching on } G[L_t \cup R]$\;
Let $e_t$ be the edge assigned to $\ell_t$ in $M_t$\;
\If{$M \cup e_t$ is a matching}{$M = M \cup e_t$\;}
}
Return $M$\;
\end{algorithm}
\vspace{10pt}
Our main result for this section is the following.

\begin{theorem}\label{thm:169}
Algorithm~\ref{alg:ebom} is $16/9$-competitive.
\end{theorem}

Let $s_t$ denote the number of elements which have arrived exactly once by round $t$.  We would like to show that $s_n$ is concentrated around it mean. 

We require The Azuma-Hoeffding inequality (see e.g., \cite{CL06}):

\begin{theorem}[Azuma-Hoeffding]
Let $Z_i, i=0,1,\ldots, n$ be a martingale such that for each $i$, $|Z_i - Z_{i-1}|< c$.
Then, for any $\lambda>0$
$$\Pr[|Z_n-Z_0|>\lambda] \leq 2 e^{ -\frac{\lambda ^2}{2nc^2} }.$$
\end{theorem}

We note that computing more precise values for the approximation bounds in Claims~\ref{claim:azuma}, \ref{claim:we} and \ref{claim:we2} is straightforward; for clarity and simplicity we use little-o notation. We use ``with high probability'' to denote ``with probability at least $1-\frac{1}{n^{\alpha}}$'', for $\alpha>1$.
\begin{claim}
\label{claim:azuma}
With high probability, $s_n = \frac{n}{2} \pm o(n).$
\end{claim}

\begin{proof}
Let $Y_i$ be a random variable whose value is $1$ if the vertex arriving on round $i$ is arriving for the first time and ($-1$) otherwise. 
Denote, for every $j\in \{0,1,\ldots n\}$, $Z_j = \expect[s_n|Y_1, Y_2, \ldots, Y_j]$. $Z_j$ is the expected number of elements that will have arrived exactly once by round $n$ given $Y_1,\ldots, Y_j$.
$Z_0$ is the expected value of $s_n$. From Claim~\ref{claim:xe} and the linearity of expectation $Z_0=  \frac{n^2}{2n-1}$.
$s_n$ satisfies the Lipschitz condition: Let $y_j$ be the realization of $Y_j$, for $j\in \{1,2,\ldots,n\}$, then
$$(s_n|y_1,\ldots, y_{i-1}, 1, y_{i+1}, \ldots, y_n) - (s_n|y_1,\ldots, y_{i-1}, -1, y_{i+1}, \ldots, y_n) = 2.$$
Therefore for all $i$, $|Z_i - Z_{i-1}| \leq 2$, and we can therefore apply the Azuma-Hoeffding inequality:
$$Pr[|Z_0 - Z_n|>c\sqrt{n}]\leq e^{\frac{-c^2}{8}}.$$
The claim follows.
\end{proof}

Claim~\ref{claim:azuma} shows that the size of the matching $M$ in Algorithm~\ref{alg:ebom} at round $n$ is approximately $n/2$. W.h.p. the number of vertices that have not yet arrived is approximately $n/4$, as  for any vertex $i$, the probability that neither of its arrivals is in the first half is $1/4 \pm o(1)$. Therefore we have the following corollary to Claim~\ref{claim:azuma} :

\begin{corollary}\label{corr:sizeOfm}
With high probability, the size of the matching $M$ at round $2n$ is at most $\frac{3n}{4} +o(n)$.
\end{corollary}

We need a few more simple claims before we can prove Theorem~\ref{thm:169}. Recall that $e_t$ is the edge matched to vertex $\ell_t$ on round $t$. Denote by $OPT$ the weight of the maximum weight matching. 
\begin{claim} For all $t$, such that a new vertex arrives at time $t$, \label{claim:we}
$$\expect[w(e_t)]\geq \frac{OPT}{n}.$$
\end{claim}

\begin{proof}
In round $t$, we can view $\ell_t$ as being selected uniformly at random from $L_t$, and so the expected weight of $e_t$ in $M_t$ is $\frac{w(M_t)}{|L_t|}$. We can view $L_t$ as being a set of size $|L_t|$  selected uniformly at random from $L$, therefore $\expect[w(M_t)]\geq OPT\cdot \frac{|L_t|}{n}$. Combining the two inequalities gives the claim.
\end{proof}

\begin{claim} \label{claim:we2} For $t > n$ ,$$\Pr[M \cup e_t \text{  is a matching}] \geq 1/4 - o(1).$$
\end{claim}

\begin{proof}
The edge $e_t = (\ell_t,r)$ can be added to $M$ only if $r$ is unmatched in $M$. From Corollary~\ref{corr:sizeOfm}, we know that w.h.p. the size of the matching $M$ at round $n$ is at most $\frac{3n}{4} +o(n)$. Therefore, on any round, at least $n/4 - o(n)$ vertices of $R$ must be unmatched. As $r$ can be seen as being selected uniformly at random from $M_t$, the claim follows.
\end{proof}

The proof of Theorem~\ref{thm:169} follows from Corollary~\ref{corr:sizeOfm} and  Claims~\ref{claim:we} and \ref{claim:we2}:

\begin{proof}[Proof of Theorem~\ref{thm:169}]
We have that for all $t\geq n$, $\Pr[M \cup e_t \text{  is a matching}] \geq 1/4 - o(1)$. There are at least $n/4 - o(n)$ vertices that have not yet arrived by round $n$. As each of them has a probability of at east $1/4-o(1)$ to be matched, by a union bound, between rounds $n$ and $2n$, we add to $M$ edges that weigh a total of $\frac{OPT}{16}-o(n)$ in expectation.
\end{proof}

\paragraph{Acknowledgments}
We would like to thank Yishay Mansour for his helpful  discussions. 

\bibliographystyle{alpha}
\bibliography{Vardi_PhD_Bibliography}
\appendix

\section{The Twice-Returning Secretary (the No Waiting Case)}
\label{app:3sec}
 We give a succinct description the algorithm of Section~\ref{sec:classic_return} with $f(n)=0$, adapted to the $2$-returning secretary problem: as the secretaries arrive, we keep note of the best secretary we have seen so far, without hiring any. Once the secretary we have marked as the best so far returns for the second time, we hire him. 
 \begin{lemma}
There is an algorithm that hires the best secretary with probability $0.9$, in the $2$-returning secretary, even without previous knowledge of $n$.
 \end{lemma}
 \begin{proof}
As before, if we reach round $i$ and see Don, we say we \emph{win} on round $i$, and denote this event $win_i$. The probability of winning on rounds $1-3$ is easy to compute directly:
\begin{align*}
Pr[win_1] &= \frac{3}{3n}, \\
Pr[win_2] &=  \left(\frac{3n-3}{3n}\right)\left(\frac{3}{3n-1}\right), \\
Pr[win_3] &=  \left(\frac{3n-3}{3n}\right)\left(\frac{3n-4}{3n-1}\right)\left(\frac{3}{3n-2}\right).
\end{align*}
However, from the fourth round on, we need to verify the best secretary so far has not yet appeared three times. This is similar to case of a single return, except that we win in round $i$ if we see Don and the best secretary so far has appeared either once or twice. These are mutually exclusive events. When analyzing the secretaries that arrived until round $i-1$, we do the following: we find the best secretary so far. If he has appeared twice, we break the tie uniformly at random, choosing one of the copies to be the best. The best secretary so far appeared in some round, $j$. For simplicity, we relabel round $j$ by $1$ and the rounds $1,2,\ldots, j-1, j+1, \ldots, i-1$ by $2, 3, \ldots, i-1$. 
\begin{align*}
Pr[win_i] &= \left(\frac{3n-3}{3n}\right)\left(\frac{3n-6}{3n-1}\right)\left(\frac{3n-7}{3n-2}\right)\dots  \left(\frac{3n-i-3}{3n-i+2}\right)\left(\frac{3}{3n-i+1}\right)\\
&+  \left(\frac{3n-3}{3n}\right) \dsum_{j=2}^{i-1} \left(\frac{2}{3n-1}\right) \left(\frac{3n-6}{3n-2}\right) \left(\frac{3n-7}{3n-3}\right)\dots  \left(\frac{3n-i-2}{3n-i+2}\right)\left(\frac{3}{3n-i+1}\right).
\end{align*} 
The first line is the probability of winning  if the best secretary only appeared once, the second is if the best secretary appeared twice already, where the sum is over all the possible locations of the repeating secretary.
This simplifies to 
\ignore{
\begin{align*}
Pr[win_i] &= \frac{3(3n-i)(3n-i-1)(3n-i-2)(3n-i-3)}{3n(3n-1)(3n-2)(3n-4)(3n-5)} +   \frac{3(2i-6)(3n-i)(3n-i-1)(3n-i-2)}{3n(3n-1)(3n-2)(3n-4)(3n-5)}\\
& =(3n-i-3 + 2i-6)\left(\frac{3(3n-i(3n-i-1)(3n-i-2)}{3n(3n-1)(3n-2)(3n-4)(3n-5)}\right)\\
& = \frac{(3n-i)(3n-i-1)(3n-i-2)(3n+i-9)}{n(3n-1)(3n-2)(3n-4)(3n-5)}.
\end{align*}
Therefore,
}
\begin{align*}
Pr[win] &= \frac{3}{3n}+  \left(\frac{3n-3}{3n}\right)\left(\frac{3}{3n-1}\right)+ \left(\frac{3n-3}{3n}\right)\left(\frac{3n-4}{3n-1}\right)\left(\frac{3}{3n-2}\right)\\
& + \dsum_{i=4}^{3n-3} \frac{(3n-i)(3n-i-1)(3n-i-2)(3n+i-9)}{n(3n-1)(3n-2)(3n-4)(3n-5)}.
\end{align*}
It is easy to verify that $Pr[win]\rightarrow 0.9$ as $n \rightarrow \infty$.
\end{proof} 
Using the methods of this section, it is easy (although cumbersome) to compute the exact probability of $Pr[win]$ for the algorithm above adapted to any $k$. In the following section, we show that if $k=\Theta(\log{n})$, we can hire the best secretary with arbitrarily high probability.

\ignore{

\begin{table}[h]
\begin{tabular}{|l|l|l|l|l|l|l|l|l|l|l|}
\hline
Returns $\setminus$ Candidates & 1 & 2   & 3    & 4  & 5 & 6 & 7  & 8 & 9 & $ \infty$  \\ \hline
$Pr[win]$, $k=1$ & 1.0	& 0.5	& 0.5	& 0.458	& 0.433	& 0.428	& 0.414	& 0.410	& 0.406 & 0.368 \\ \hline
$Pr[win]$, $k=2$  & 1.0 & 0.833 & 0.778 & 0.75 & 0.733 & 0.722 & 0.714 &  0.708 & 0.704 & 0.667 \\ \hline
$Pr[win]$, $k=3$  & 1.0 & 0.95 & 0.880 & 0.864 & 0.861 & 0.862 & 0.864 &  0.866 & 0.868 & 0.9 \\ \hline
\end{tabular}
\caption {Comparison of optimal strategy for the calssical secretary problem ($k=1$) with the no waiting strategy for the $1$- and $2$-returning secretary problem, for several values of $n$.} \label{table1}
\end{table}

}

\section{The $\Theta(\log{n})$-returning secretary}
\label{app:logn}
In this section, we show that in the $k$-returning secretary problem, if $k=\Theta(\log{n})$, we can guarantee an arbitrarily high probability of success, as stated in the following theorem.

\begin{theorem}
For every $\alpha \in \mbR^+$, there exists a constant $c>0$, such that in the $c\log{n}$-returning secretary problem,
$$Pr[win]\geq 1-\frac{1}{n^{\alpha}}.$$
\end{theorem}
\begin{proof}
Assume each secretary returns $k= c \log{n}$ times, for $c>0$ to be determined later. 
Denote by $t_a^i \in \{1,2,\ldots,kn\}$ the time of secretary $a$'s $i^{th}$ arrival. Let $X_{a,b}$ be a random variable whose value is $1$ if $\forall i,j t_a^i < t_b^j$ and $0$ otherwise. That is, $X_{a,b}=1$ iff all of the $k$ appearances of secretary $a$ occur before all of the $k$ appearances of secretary $b$.
 We show that for any two secretaries, $a$ and $b$, 
$$Pr[X_{a,b}=1]\leq 1/n^{\alpha}.$$

Take an arbitrary ordering on the $n-2$ secretaries that are not $a$ or $b$ and fix the $2k$ possible positions of $a$ and $b$ relative to this ordering. There are $(2k)!$ possibilities for the appearances of the two secretaries. As the probability of $X_{a,b}$ is independent of this ordering, 
\begin{equation*}
Pr[X_{a,b}] = \frac{((c\log{n})!)^2}{(2c\log{n})!} = {{2c\log{n} \choose c\log{n}}}^{-1} \leq 2^{-c\log{n}}.
\end{equation*}

Denote the best secretary by $s$. Taking a union bound over $Pr[X_{i,s}], i\neq s$, and choosing an appropriate value for $c$, completes the proof.
\end{proof}

\end{document}